\documentclass[a4paper, 11pt]{article}
\usepackage{tikz-cd}
\usepackage{geometry}
\usepackage{cancel}
\usepackage{empheq}
\usepackage{mathtools}
\usepackage{amsthm}
\usepackage{amsfonts}
\usepackage[backend=biber]{biblatex}
\usetikzlibrary{calc, arrows.meta}

\def\R{\mathbb{R}}
\def\C{\mathbb{C}}
\def\A{\mathbb{A}}
\def\F{\mathbb{F}}
\def\L{\mathbb{L}}

\def\K{\mathbb{K}}

\def\al{\alpha}
\def\be{\beta}
\def\ga{\gamma}
\def\de{\delta}
\def\ep{\epsilon}
\def\ze{\zeta}

\def\io{\iota}
\def\ka{\kappa}

\def\si{\sigma}

\def\om{\omega}

\def\La{\Lambda}
\def\Si{\Sigma}

\def\na{\nabla}
\def\pa{\partial}

\def\_#1{{}^{\boldsymbol{\cdot}}_{#1}{}}
\def\^#1{{}_{\, \boldsymbol{\cdot} \,}^{#1}{}}

\newtheorem{lma}[equation]{Lemma}
\newtheorem{corollary}[equation]{Corollary}

\makeatletter

\def\Ri{\futurelet\next\@Ri}

\def\@Ri{%
  \ifx\next(%
    \let\@RiDo\@RiOpt
  \else
    \let\@RiDo\@RiNoOpt
  \fi
  \@RiDo
}

\def\@RiNoOpt#1_#2.{%
  R^{\,#1}{}_{#2}%
}

\def\@RiOpt(#1)#2_#3.{%
  {}^{#1}\!\@RiNoOpt#2_#3.%
}

\makeatother

\def\Rs(#1){\Ri(#1)_.}

\def\we{\wedge}
\def\tp#1{{}^\ast\!#1}
\def\np#1{{}_\circ#1}

\def\SU(#1){\mathrm{SU}(#1)}
\def\SL(#1){\mathrm{SL}(#1)}
\def\SO(#1){\mathrm{SO}(#1)}
\def\sl(#1){\mathfrak{spin}(#1)}
\def\su(#1){\mathfrak{su}(#1)}
\def\spin(#1){\mathfrak{spin}(#1)}
\def\Spin(#1){\mathrm{Spin}(#1)}

\def\dunderline#1{\underline{\underline{#1}}}

\def\chr#1#2#3#4{\{#1\}^#2_{#3 #4}}

\def\vec#1{\vbox{\offinterlineskip\ialign{%
    \hfil##\hfil\crcr %
    $\scriptscriptstyle\rightharpoonup$\crcr 
    \noalign{\kern 1pt}
    $#1$\crcr
}}}

\def\cvec#1{\vtop{\offinterlineskip\ialign{%
    \hfil##\hfil\crcr
    $#1$\crcr
    \noalign{\kern 1pt}
    \hidewidth$\scriptscriptstyle\rightharpoondown$\hidewidth\crcr
}}}

\makeatletter
\renewcommand{\author}[2][]{%
  \ifx\@authors\@undefined
    \def\@authors{#2\textsuperscript{\it #1 \/ }}
  \else
    \g@addto@macro\@authors{\quad #2\textsuperscript{\it #1 \/ }}
  \fi
}

\newcommand{\affiliation}[2][]{%
  \ifx\@affils\@undefined
    \def\@affils{\textsuperscript{\it #1\/}\ {\small #2}}
  \else
    \g@addto@macro\@affils{\\ \textsuperscript{\it #1 \/}\ {\small #2}}
  \fi
}

\patchcmd{\@maketitle}
  {\begin{center}}
  {\begin{center}
   \let\footnote\thanks}
  {}{}

\renewcommand{\@author}{%
  \begin{tabular}[t]{c}
    \@authors \\\\
    \small \@affils
  \end{tabular}
}
\makeatother

\newif\ifhide

\def\eatrow#1{}
\def\showrow#1{& #1\\}

\def\Hide#1{\ifhide
  \expandafter\eatrow
\else
  \expandafter\showrow
\fi{#1}}

\title{A Lagrangian framework for canonical analysis \\ for the Holst model with $\be = 0$}
\date{}
\author[a]{Roberto Ciccarelli}
\author[a,b]{Lorenzo Fatibene}
\affiliation[a]{Department of Mathematics {\it ``Giuseppe Peano''}, University of Torino (Italy)}
\affiliation[b]{Istituto Nazionale di Fisica Nucleare (INFN), Sezione di Torino (Italy)}
\numberwithin{equation}{section}

\hidefalse

\begin{document}
\maketitle

\begin{abstract}
We perform a canonical analysis of the Holst model for General Relativity, within the framework laid out in \cite{fatibeneLectureNotes1, orizzonteBarberoImmirziConnectionsHow2021}, distinguishing our approach by setting the Barbero parameter to $\beta=0$ and leaving the lapse and shift functions unconstrained. The $\be = 0$ choice is of particular interest because it is viable across all dimensions, providing a necessary foundation for extending the Loop Quantum Gravity formalism beyond $3+1$ dimensions. 

Through field decomposition and the projection of the field equations, we derive a system of 37 equations (10 {\it differential constraints}, 21 {\it algebraic constraints}, and 6 {\it evolution equations}) exactly matching the 37 field components to be determined. Moreover, leaving the gauge unfixed reveals that three equations, which are typically identically satisfied under normal evolution, are actually differential constraints whose triviality depends on specific gauge choices. The resulting framework remains fully consistent with the standard $3+1$ decomposition of the Einstein equations without requiring any constraints on the lapse and shift functions.
\end{abstract}

\section{Introduction}
General Relativity (GR) is the (classical) theory of gravitation, that is based on the idea that the gravitational attraction arises as a consequence of the geometry of spacetime.
In this framework, spacetime is described by a manifold $M$ and  the gravitational field by a Lorentzian metric on it. It is however well known that the metric field is not the only possible way to describe the geometric properties of a manifold. An alternative approach is the so called \textit{frame-affine} formulation of General Relativity, where the metric is replaced by a frame field and an independent spin connection. 

The Holst model is a modification of the frame-affine formulation of General Relativity, that is the starting point for the Loop Quantum Gravity approach. 
It is based on the introduction of an extra term in the action, coupled with a dimensionless {\it Holst parameter\/ } $\ga\not=0$ so that the resulting model is dynamically equivalent to standard GR for any value of $\ga$. For a broad discussion on the Loop Quantum Gravity, see \cite{rovelliCovariantLoopQuantum2015, rovelli2004quantum,Ashtekar_2021, Thiemann:2002nj,Thiemann:2007pyv,AbhayAshtekar2004, Thiemann2003} and references quoted therein. Canonical analysis of the Holst model is performed also in \cite{geillerNoteHolstAction2013, montesinos}.
In \cite{fatibeneLectureNotes2, fatibeneLectureNote3} the notation and the geometric setting of the theory are laid out in detail, while in \cite{Fatibene_2007a, Fatibene_2007, Fatibene:2012px, orizzonteBarberoImmirziConnectionsHow2021} a discussion on the existance of the Barbero-Immirzi connection on spacetime is given.

In this work we are interested in performing the canonical analysis of the Holst model with  the {\it Barbero parameter} $\be = 0$, which appears 
when one introduces the Ashtekar-Barbero-Immirzi variables. 

Often in the literature one finds $\ga=\be$ (see \cite{physics4040072} for an historical review). We consider the two parameters unconstrained instead, for a couple of different reasons.
First of all, the Holst parameter is dynamical, while the Barbero parameter is kinematical, and it would be a bit suspicious to set them equal.
Moreover, while the Holst parameter is purely related to dimension $m=4$, (we want to check here that) the value $\be=0$ is viable in all dimensions, 
even though non zero values are available in dimension $m=4$, only.

Hereafter, we set $m=4$ and we are interested in checking that we can cleanly obtain the constraints and evolution equations without any restriction on the value of the Holst parameter $\ga$, which is a dynamical parameter, and because it is necessary if one wants to extend the LQG formalism in dimension different from $3+1$.
We agree that extending LQG to dimensions different from $m=4$ is a purely mathematical issue, still we know many things go through extending standard GR to different dimensions and signatures. We believe that this would highlight the structure of the theory; see also \cite{Fatibene:2012px}.

In the first section we will briefly introduce the Holst model and the Ashtekar-Barbero-Immirzi variables on spacetime, instead of on space only, by a canonical transformation, as one usually does. Moreover, our analysis is completely Lagrangian and because of that it may give insight to the spin foam formulation in future investigations.
In the first section we will also re-write the field equations in terms of Ashtekar-Immirzi variables. 
We shall also discuss global topological restrictions for the formalism to be viable. Even though in dimension $m=4$ many bundles are constrained to be trivial, the bundles structure is needed to single out transformations that are relevant, e.g.~$\SU(2)$-gauge transformations.

In the second section we will perform the canonical analysis of the field equations, obtaining the constraint and evolution equations; \cite{Fatibene:2015ita,Fatibene:2016usr}. Finally, in the last section we will draw some conclusions and outline some future perspectives.

We add some appendices to present technical results and keep the analysis clean.

\section{Holst model}

Let us first discuss topological conditions and constructions. If one is interested in a local formulation, this part can be skipped and one can jump ahead to subsection \ref{sec:choosing_beta}.

Let us start from a $4-$dimensional manifold $M$, called the \textit{spacetime}, and a principal bundle $[P \xrightarrow{\pi} M]$ with structure group $\Spin(3,1)\simeq \SL(2, \C)$. The variables of Holst theory are a spin frame $e_a = e_a^\mu \pa_\mu$, defined as a global map $e: P \to L(M)$ vertical and equivariant with respect to the double covering map $\ell: \Spin(3,1) \to \SO(3,1)$, and a $\Spin(3,1)$ principal connection $\om = dx^\mu \otimes \left(\pa_\mu - \hat \om^{ab}_\mu\si_{ab}\right)$ on $P$, where $\si_{ab}$ is a basis of vertical right invariant vector fields on $P$. For more details on the geometric setting of the theory, see \cite{fatibeneLectureNotes1}.
In particular, for this construction to be possible, we need that $M$ has vanishing first and second Stiefel-Whitney classes, so that it also allows global Lorentzian metrics and eventually spacetime allows global spinor fields.

In this setting, the Hilbert-Einstein action in vacuum with cosmological constant $\La$ can be modified with the addition of a term, depending on the {\it Holst parameter} $\ga$, to obtain the so called \textit{Holst action}, which is the starting point for the Loop Quantum Gravity approach to quantum gravity. The Lagrangian of the Holst model is
\begin{equation}
L _\text{Holst}= \tfrac 1 {4\ka} \left[ (\hat R^{ab} \we e^c \we e^d)\ep_{abcd} + \tfrac 2 \ga \hat R^{ab} \we e_a \we e_b - \tfrac \La 4 \ep_{abcd} e^a\we e^b \we e^c \we e^d \right],
\label{eq:holst_lagrangian}
\end{equation}  
where $\ga \in \R \setminus \{0\}$ is the \textit{Holst parameter}.
Here $e^c = e^c_\mu dx^\mu$ is the spin coframe (which exist global if $M$ is a {\it spin manifold}, i.e.~its 1st and 2nd Stiefel-Whitney class are zero), $\hat R^{ab} = \tfrac 1 2 \hat R^{ab}{}_{\mu\nu}dx^\mu \we dx^\nu$ is the curvature of $\hat \om^{ab}$ and $\hat \na$ is the covariant derivative induced by $\hat \om^{ab}$. Latin lowercase indices of the first part of the alphabet $a, b, \ldots$ are indices on $\Spin(3,1)$ and run from $0$ to $3$, while greek indices are indices on spacetime and run from $0$ to $3$. Latin indices in the central part of the alphabet $i, j, \ldots$ are indices on $\SU(2)$ hence ranging on $1,\ldots, 3$,
while the first capital Latin $A, B$ alphabet are indices  on space and also running from $1$ to $3$.

Let us stress that $\hat \om^{ab}_\mu$ is a connection that, a priori, is chosen independently from the spin connection induced by the spin frame $e_a$, i.e. $\tilde \om^{ab}_\mu = e^a_\al(\chr g\al\be\mu e^{b\be} + d_\mu e^{b\al})$, where $g = e^a_\al\eta_{ab}e^b_\nu dx^\mu \otimes dx^\nu$ is the metric induced by the frame. 

This theory is known to be equivalent to standard General Relativity for every value of the parameter $\gamma \neq 0$, and the field equations obtained via the variational principle are
\begin{equation}
    \begin{dcases}
    \hat\na  B_{ab} = 0\\
    \hat R^{ab} \we \left(\ep_{abcd} e^c  + \tfrac 2 \ga e_{[a} \eta_{b]d}\right) = \tfrac \La 3 \ep_{abcd} e^a \we e^b \we e^c 
\end{dcases}
\label{eq:holst_equations}
\end{equation}
where $B_{ab} = e^c \we e^d \ep_{abcd} + \tfrac 2 \ga e_a \we e_b$ for brevity. The first equation is algebraic and linear in the $\Spin(3,1)$-connection and implies that $\hat \om^{ab} = \tilde \om^{ab}$, while the second one is the Einstein equation in the frame-affine formulation of General Relativity.

\subsection{Ashtekar-Barbero-Immirzi variables on spacetime}
The group $\Spin(3,1)\simeq\SL(2,\C)$ is not compact, and a compact structure group is desirable when dealing with quantisation; however, $\Spin(3,1)$ has a compact subgroup, namely $\Spin(3) \simeq \SU(2)$ that can be used to construct a \textit{reduction} of $P$ to $\SU(2)$, that is a $\SU(2)$-principal bundle $\Si \xrightarrow{\ \tau \ } M$ together with a map $\io: \Si \to P$  vertical and equivariant with respect to the group homomorphism $i : \SU(2) \to \SL(2,\C)$, that is $\io(p\cdot U) = \io(p) \cdot i(U)\ \forall \, p \in \Si, U \in \SU(2) $. Such a reduction is, in general, not guaranteed to exist due to topological obstructions; however a reduction from $\Spin(n,1)$ to $\Spin(n, 0)$ is proven to exist, provided that the spacetime is a spin manifold,  for every dimension $n$ (see \cite{orizzonteBarberoImmirziConnectionsHow2021}).

The existence of the reduction allows us to define the $\SU(2)$ spin frame $\ep^\mu_a$ as $\ep = e \circ \io$. However, see \cite{orizzonteBarberoImmirziConnectionsHow2021}, it is not sufficient to deal with the principal connection $\hat \om$ and obtain a $\SU(2)$ connection, because the horizontal subspaces defined by $\hat \om$ might not be tangent to $\io(\Si)$. This fact is assured in case a \textit{reductive splitting} exists. For the purpose of this work, it is sufficient to state that a reductive splitting is a map $\Phi: \spin(n,1) / \spin(n, 0) \to \spin(n,1)$ such that its image in $\spin(n,1)$ is invariant with respect to the adjoint action of $\Spin(n,1)$ restricted to $\Spin(n,0)$. For a detailed treatment, see \cite{kobayashiFoundationsDifferentialGeometry1996}. Indeed, for $n = 3$ there exists a whole family of reductive splittings $\Phi_\beta$,  while if $n > 3$ a single one, $\Phi_0$, is available. 

The maps $\Phi_\beta$ allow us to split $\hat\om$ into a $\SU(2)$ connection $\A = dx^\mu \otimes \left(\pa_\mu - \A^i_\mu \si_i\right)$, called the \textit{Barbero-Immirzi} connection, and a $\su(2)$-valued $1-$form $k = k^i_\mu dx^\mu \otimes \si_i$, the \textit{Immirzi tensor}. These new variables, that are the Ashtekar-Barbero-Immirzi variables on spacetime, are defined as follows 
\begin{equation}
    %\begin{aligned}& 
    \A^i_\mu = \tfrac 1 2 \ep^i{}_{jk}\hat\om^{jk}_\mu + \be \hat \om^{0i}_\mu%\\
    \qquad\qquad%& 
   k^i_\mu = \hat \om^{0i}_\mu
	%\end{aligned}
\label{eq:change_om}
\end{equation} 

Let us remark that, following \cite{Fatibene_2007,Fatibene_2007a,orizzonteBarberoImmirziConnectionsHow2021} we define  {\it Ashetekar-Barbero-Immirzi (ABI) connection} on spacetime. When we restrict to space we obtain exactly the ABI connection. The reduction $\Si\to P$ is essential to define $\SU(2)$-gauge transformations and show these are $\SU(2)$-fields, 
in particular $\A^i_\mu$ is an $\SU(2)$-connection.

\subsection{Choosing values for $\be$}
\label{sec:choosing_beta}

In literature, the value of $\be$ is chosen to be $\be = \ga$. As outlined in \cite{fatibeneLectureNotes1}, this is an awkward choice, since the parameter $\gamma$ is a dynamical one --- changing it implies changing the action of the theory, even though it is still equivalent to General Relativity --- while $\be$ is kinematic, it is just parametrising the algebraic variable change $\hat \om^{ab}_\mu \mapsto \left(\A^i_\mu, k^i_\mu\right)$. In \cite{fatibeneLectureNotes1} is then shown how to obtain the equations that are the starting point for the Loop Quantum Gravity approach in dimension $m = 4$ for a generic $\be \neq 0 \neq  \ga$. 

The case $\be = 0$ is often overlooked, although it is interesting for mainly two reasons. First, the prescription $\be \neq 0$ is no more necessary\footnote{The condition $\gamma \neq 0$ is necessary because, when taking $\ga \to 0$, the first term in Holst Lagrangian becomes negligible and the dynamical equivalence with standard General Relativity is lost.} when $\be \neq \ga$. 
The second reason is that in dimension $m \neq 4$ it is still possible to define a Holst action with no Holst term (which corresponds to the limit $\ga\rightarrow \infty$) and reduce fields to $\Spin(n)$ objects. However, the only reductive splitting allowed is $\Phi_0$. It is then clear that the study of the case $\be = 0$ is necessary if one wants to discuss and extend the LQG formalism in any dimension other than $3+1$.

\subsection{Field equations in terms of Ashtekar-Barbero-Immirzi variables}
Let us then choose $\be = 0$ and re-write the field equations \eqref{eq:holst_equations} in terms of the new variables $\A^i_\mu$ and $k^i_\mu$. These equations can be obtained both by substituting \eqref{eq:change_om} in \eqref{eq:holst_lagrangian} and by applying the variational principle, or by substituting \eqref{eq:change_om} directly in \eqref{eq:holst_equations}.

First of all, consider  \eqref{eq:change_om} and its inverse with $\be = 0$
\begin{equation}
    \begin{aligned}
    &\A^i_\mu = \tfrac 1 2 \ep^i{}_{jk} \hat \om^{jk}_\mu  \qquad && \hat \om^{jk}_\mu = \tfrac 1 2\ep_{i}{}^{jk} \A^i_\mu \\
    & k^i = k^i_\mu dx^\mu && k^i_\mu = \om^{0i}_\mu
    \end{aligned}
    \label{eq:variables_change}
\end{equation}

The $\SU(2)$-connection $\A^i$ defines a covariant derivative $\na a^i = d a^i + \ep^i{}_{lm}\A^m \we a^l$ and its curvature is
\begin{equation}
\F^i = \tfrac 1 2 \F^i_{\mu\nu} dx^\mu \we dx^\nu = \tfrac 1 2 \left(d_\mu \A^i_\nu - d_\nu \A^i_\mu - \ep^i{}_{jk} \A^j_\mu \A^k_\nu\right) dx^\mu \we dx^\nu
\end{equation}
With straightforward substitution, re-write the curvature $\hat R^{ab}$ of $\hat\om^{ab}_\mu$ with respect to $k^i$ and $\F^i$
\begin{align}
    \hat R^{0i} &= \na k^i \\ 
    \hat R^{ij} &= \tfrac 1 2 \F^k\ep_k{}^{ij} + k^i \we k^j
\end{align}
We can also define also the $2-$forms valued in $\spin(3)$
\begin{align}
    \K_i &= \tfrac 1 2 B_{0i} = \left(\tfrac 1 2 \ep_{ijk} e^j \we e^k - \tfrac 1 \ga e^0 \we e_i\right) \\ 
    \L_k &= \tfrac 1 2 \ep^{ij}{}_k B_{ij} = e^0 \we e_k + \tfrac 1 {2\ga}\ep_{klm} e^l \we e^m
\end{align}

Then consider the $(0i)$ and the $(ij)$ component of the two equations in \eqref{eq:holst_equations} and, up to a linear combination of the equations, the field equations can be re-written as
\begin{subequations} 
\begin{empheq}[left=\empheqlbrace]{align}
              &\na \left(\L^k + \ga \K^k\right) = \ep^k{}_{lm} k^l \we \left(\K^m - \ga \L^m\right) \label{eq:field_1}\\ 
          &\na \left(\K^k - \ga \L^k \right) = - \ep^k{}_{lm} k^l \we \left(\L^m + \ga \K^m\right) \label{eq:field_2}\\
          &\F_k\we e^k - \tfrac 1 \ga \na k^k \we e_k + \tfrac 1 2 \ep_{kij} k^i \we k^j \we e^k = \tfrac \La 6 \ep_{ijkl} e^i \we e^j \we e^k \label{eq:field_3}\\
          &\begin{aligned}[t]
            &\left(\tfrac 1 \ga \ep^h{}_{ij}\F^i + e^h{}_{ij}\na k^i  - \tfrac 1 \ga k^h \we k_j \right)  \we e^j = \\ 
          &\qquad = - \left( \F^h - \tfrac 1 \ga \na k^h + \tfrac 1 2 \ep^h{}_{ij} k^i \we k^j - \tfrac \La 2 \ep^h{}_{ij}e^i \we e^j\right) \we e^0  \end{aligned}\label{eq:field_4}
    \end{empheq}
          \label{eq:field_spacetime}
\end{subequations}

\vspace{0.1cm}

These equations are $(3\cdot 4) + (3\cdot 4)  + (1\cdot 4) + (3\cdot 4) = 40$ equations in the fields $(e^0_\mu, e^i_\mu, \A^i_\mu, k^i_\mu)$, that is $ 4 + (3\cdot 4) + (3\cdot 4) + (3 \cdot 4) = 40$ unknowns.

\section{Canonical analysis of the field equations}
In order to obtain the constraints and evolution equations, we are now interested in defining a Cauchy problem for the field equations (see \cite{fatibeneLectureNotes2}). 

Consider a $4-$dimensional compact subset $\bar D \subset M$ and a vector field $\ze$ on $\bar D$ that is non-zero on $D := \mathrm{Int} (\bar D)$ and vanishing on the boundary $\pa D$. If the set of its integral curves forms a manifold, the induced foliation of $D$ defines a trivial principal bundle $[D \to S]$. Then, fix a section $S_0 \in D$ and drag it along the flow of $\ze$ to obtain a regular foliation $S_t \xhookrightarrow{\ i_t\ } D$. 
The section $S_0$ is called the Cauchy surface, together with the values of the fields restricted on it, it defines an \textit{initial condition} for the Cauchy problem.  

\subsection{Decomposition of the fields}
Consider the canonical covector $\cvec u$ defined as $\cvec u(v) = 0$ if and only if a vector $v$ is tangent to the leaves. Once a metric $g$ is defined, we can consider the unit normal covector $\cvec n$ and the unit normal vector $\vec n$ so that $g(\vec n, \vec n) = -1$ and $\cvec n(\vec n) = 1$. Choose now coordinates $(t, x^A)$ adapted to the foliation, so that $\ze = \pa_0$ so that two bases are induced on $T_xD$, namely $(\pa_0, \pa_A)$ and $(\vec n, \pa_A) $. The change of basis is given by 
\begin{equation}
    \pa_0 = N \vec n + \be^A \pa_A
\end{equation} 
where $N$ is the {\it lapse function}, $\be^A\pa_A$ is the {\it shift vector} and $\bar N = 1/N$.

\paragraph{Decomposing tensor fields.} The normal vector and covector are then expressed in adapted coordinates as 
\begin{equation}
    \vec n = \bar N \left(\pa_0 - \be^A \pa_A\right) = \bar N \vec m, \qquad \cvec n = - N dt.
\end{equation}
and define a projection of any tensor field onto tangents and normal directions. In general, a $(r,s)-$ tensor field will need to be projected $r+s$ times, but since in the theory we only have vectors $v$ and forms $\al$, we will have only two projections for each field, the \textit{tangent projection} $\tp(\cdot)$ and the \textit{normal projection} $\np(\cdot)$, which can be defined as follows
\begin{equation}
      \tp \al = (i_t)^* \al,  \quad\quad \np\al = \al(\cvec n), \quad\quad 
      \tp v = v + \cvec n(v)\vec n,   \quad\quad \np v = \cvec n (v).  
\end{equation}
In adapted coordinates we can express the projections as 
\begin{subequations}
    \renewcommand{\theequation}{\theparentequation\alph{equation}}
    \begin{align}
        &\tp\al = \al_A dx^A, && \np \al  = \bar N\left(\al_0 - \be^A \al_A\right). \label{eq:projections_alpha} \\ 
        &\tp v = v^A \pa_A, && \np v = - N v^0. \label{eq:projections_v}
    \end{align}
    \label{eq:projections}
\end{subequations}
A tensor that occupies a relevant role in the theory is the metric $g_{\mu\nu} := e^a_\mu \eta_{ab} e^b_\nu$ induced by the frame. Since it is closely related to the definition of the normal vector $\vec n$, its projections (together with the ones of its inverse) are easily computed to be
\begin{subequations}
\begin{align}
    & g_{00} = - N^2 + \ga_{AB} \be^A \be^B    && g_{0A} = \ga_{AB} \be^B  && g_{AB} = \ga_{AB} \\     
    & g^{00} = -\bar N^2  && g^{0A} = \bar N^2 \be^A  && g^{AB} = \ga^{AB} - \bar N^2 \be^A\be^B.
\end{align}
\label{eq:projections_metric}
\end{subequations}

\paragraph{Decomposing the frame.} Thanks to gauge transformations of $\Spin(3,1)$, it is always possible to \textit{adapt} the spin frame $e_a$ so that $e_0 = \vec n$ and $e_i$ is tangent to $S_t$. Therefore a \textit{triad}  $\ep_i = \tp e_i$ is defined and the projections of the frame are simply given by 
\begin{equation}
    \tp(e_0) = 0 \qquad \np(e_0) = 1  \qquad \tp (e_i) = \ep_i \qquad  \np(e_i)  = 0.
    \label{eq:projections_e}
\end{equation}

The $16$ degrees of freedom of $e_i^\mu$ are then comprised in the triad $\ep_i^A$ ($9$ components), in the choice of the normal vector $\vec n$ ($4$ components) and in the $3$ boosts that we used to adapt the frame. Computations involving $de^i$ and some expressions of the frames, needed to decompose the equations, are exposed in Appendix \ref{appendix:projections}.

\paragraph{Auxiliary fields.} 
The connection $\hat\om$ is not the only connection that we can build from the fields of theory: in fact, the adapted tetrad induces a Levi Civita connection $\tilde \om^{ab}$. Since we know that in the frame-affine formulation of General Relativity the equality $\hat \om^{ab} = \tilde \om^{ab}$ arises as an algebraic consequence of the field equations, we can consider the difference of the Barbero connection and Immirzi form of $\hat \om$ to the ones of $\tilde \om$. We can define therefore 
\begin{equation}
 z^i = \A^i - \tilde \A^i, \qquad \quad h^i = k^i - \tilde k^i
\end{equation}
since we expect that the field equations will imply $z^i = 0$ and $h^i = 0$ and the Levi-Civita connection simplifies the computations, being torsionless and frame compatible, and can be expressed as a function of the triad. The projections of $\A^i, k^i, z^i $ and $h^i$, together with their covariant derivatives, can be computed using \eqref{eq:projections} and can be found in \ref{appendix:projections}. 

Moreover, denoting with $\ep$ the determinant of the triad, that is
\begin{equation}
\ep = \det(\ep_i^A) = \tfrac 1 {3!} \ep_{ijk} \ep^i_A \ep^j_B \ep^k_C \ep^{ABC},
\label{eq:detepsilon}
\end{equation}
we can define the densitised triad as
\begin{equation}
    E_i = \tfrac 1 2 \ep_{ijk} e^j \we e^k.
\end{equation}

Finally, we need an \textit{evolution operator} that controls the time evolution of the fields. We will denote it as $\de_m$ and will be defined case by case, and we will denote by $D_A(\cdot)$ the covariant derivative with respect to $\A^i$ on $S$, while with $d_\ast(\cdot)$ the differential of forms on $S$.

\paragraph{Splitting equations.} Our aim is, then, to write the tangent and normal projections of the equations in \eqref{eq:field_spacetime} in terms of the metric $\ga_{AB}$, the curvature of the Barbero connection $F_{AB}$, the $3$-curvature ${}^3R^{AB}$, the extrinsic curvature $\chi_{AB}$, the auxiliary fields $(E_i, \tp z^i, \np z^i, \tp h^i, \np h^i)$ and their covariant derivatives with respect to $\A^i$ on $S$. 

The first two equations in \eqref{eq:field_spacetime} can be treated separately from the last two, since the former do not depend on the Immirzi form $k^i$ and the curvature of $\A^i$, while the latter do not depend on $\K^i$ and $\L^i$. Moreover, from the equations in metric-affine formulation of General Relativity, we know that the first two equations will give us some algebraic constraints that will allow us to simplify the computations for the last two equations. We expect to obtain $\tp z^i_A = 0, \np z^i = 0, \tp h^i_A = 0, \np h^i = 0$ from algebraic constraints.

We will then start by treating the first two equations and then we will move to the last two, using the results obtained from the first two equations to simplify the computations.

\subsection{Projecting the first two equations}

\subsubsection{Tangent projection of the first equation: Gauss constraint}

The tangent projection of \eqref{eq:field_1} gives a differential constraint, known as the Gauss constraint. Start from 
\begin{equation*}
   \tp\Big(\na\big(\L^k + \ga \K^k\big)\Big) = \ep^k{}_{ij} \tp k^i \we \tp\left(\K^j - \ga \L^j\right)
\end{equation*}
and, noticing that the right hand side is zero because of \eqref{eq:KL_tangent}, we can rewrite the left hand side using  \eqref{eq:naLK_tangent} and obtain
\begin{equation}
    \boxed{D E_k = 0}  \label{eq:gauss_constraint}.
\end{equation}

Moreover the following corollary holds:
\begin{corollary}
    \label{cor:skew_z}
    The skew-symmetric part of $\tp z_{AB}$ is zero
    \[
        \tp z_{[AB]} = 0
    \]
\end{corollary}
\begin{proof}
From \eqref{eq:gauss_constraint},
\[
    0 = D_A E^A_k =  D_A\left(\ep \ep^A_k\right) = \ep^A_k D_A\ep + \ep D_A\ep^A_k = \ep \ep^A_k  \ep^B_l D_A \ep^l_B + \ep D_A \ep^A_k 
\]
Contracting with $\ep^C_k$ and using  \ref{lma:covder_ep} we have
\begin{align*}
    0 &= \ep^k_C \Big( \ep^A_k \ep^B_l D_A \ep^l_B + \ep^k_C D_A \ep^A_k\Big) = \ep^A_k  \ep^k_F  \left(\ep^F_{\> \cdot  \> CG}\ep^G_l \tp z^l_A - \ep^F_{\> \cdot  \> AG} \ep^G_l \tp z^l_C\right) =  \\ 
&= \ep_{BCG} \gamma^{AB} \ep^G_l\tp z^l_A = \ep_{BCG} \tp z^{GB}
\end{align*}
and thus $\tp z_{[AB]} = 0$.

\end{proof}

\subsubsection{Tangent projection of the second equation}
The tangent projection of \eqref{eq:field_2} is 
\begin{align*}
    \tp\Big(\na\big(\K^k &- \ga \L^k\big)\Big)= -\ep^k{}_{ij} \tp k^i \we \tp \left(\L^j + \ga \K^j\right).
\end{align*}
The left-hand side is zero because of \eqref{eq:naKL_tangent}, while using \eqref{eq:LK_tangent} and \eqref{eq:projections_k} we can rewrite the right-hand side as
\begin{align*}
    0 &= \ep^k{}_{ij} \tp k^i \we \left(\tfrac 1 2 \ep^j{}_{lm}\ep^l \we \ep^m\right) = \tp k^i \we \ep_i \we \ep^k  = \tp h^i \we \ep_i \we \ep^k + \chi_{ij}\ep^i \we \ep^j \we \ep^k  \\ &= \tp h^i \we \ep_i  = \tp h^i_A \ep_{iB} dx^A \we dx^B = \tp h_{AB}  dx^A \we dx^B.
\end{align*}
This implies that the skew-symmetric part of $\tp h_{AB}$ is zero, which is an algebraic constraint 
\begin{equation}
      \boxed{\tp h_{[AB]} = 0.}
      \label{eq:skew_h}
\end{equation}

\subsubsection{Normal projection of the second equation}
The normal projection of \eqref{eq:field_2}, since $\np(\L^k + \ga \K^k) = 0$ by 
\[\np\Big(\na\big(\K^k - \ga \L^k\big)\Big)= -\ep^k{}_{ij} \, \np k^i \we \tp \left(\L^j + \ga \K^j\right) - \ep^k{}_{ij}\tp k^i \we \np\left(\L^j + \ga\K^j\right) \]

Using \eqref{eq:LK_normal}, \eqref{eq:naKL_normal}, \eqref{eq:LK_tangent} and \eqref{eq:projections_k}, we can rewrite it as
\begin{align*}
   \bar N d_\ast N \we \ep^k + D\ep^k &= \ep^k{}_{ij} \np k^i \we \left(\tfrac 1 2 \ep^j{}_{lm} \ep^l \we \ep^m\right) \\ 
     &= \np k^i \we \ep_i \we \ep^k \\ 
     &= \np h^i \we \ep_i \we \ep^k + \bar N d_\ast N \we \ep^k
\end{align*}

Hence, expressing it in adapted coordinates and using Lemma \ref{lma:covder_ep}, we have
\begin{align}
    \ep^E_k \ep^{BCA} \np h^i \ep_{iB} \ep^k_C &=  \ep^E_k \ep^{BCA}D_B \ep^k_C   \nonumber \\
\bar \ep \ep^{BEA} \np h^i \ep_{iB} &= \tp z^{AE} - \tp z \ga^{EA} \label{eq:normal_second_eq}
\end{align}

Let us now consider separately the symmetric and skew-symmetric parts. The symmetric part of \eqref{eq:normal_second_eq} reads
\[
    \tp z^{(AE)} = \tp z \ga^{EA}
\]
while it trace is 
\[\tp z = 0\]
Substituting back into the symmetric part, we obtain
\begin{equation}
   \boxed{ \tp z^{(AE)} = 0 }
    \label{eq:symm_z}
\end{equation}

The skew part of \eqref{eq:normal_second_eq} is
\[
{\tp z^{[AE]}  =  \bar \ep\ep^{BEA} \np h^i \ep_{iB}}
\]
and using Corollary \ref{cor:skew_z} we  obtain 
\begin{equation}
    \boxed{\np h^i = 0}
    \label{eq:np_h}
\end{equation}

% Note that  Equation \eqref{eq:symm_z} and  Corollary \ref{cor:skew_z} imply that
% \begin{equation}
% \tp z^{AB} = 0
% \end{equation} 
% and hence Lemma \ref{lma:covder_ep} implies 
% \begin{equation}
%     D_B \ep^k_C = 0.
%     \label{eq:covder_ep_zero}
% \end{equation}

\subsubsection{Normal projection of the first equation}
The normal projection of the first equation in \eqref{eq:field_spacetime} is
\begin{equation*}
    \np\Big(\na\big(\L^k + \ga \K^k\big)\Big)= -\ep^k{}_{ij} \tp k^i \we \np\left(\K^j - \ga \L^j\right) - \ep^k{}_{ij} \, \np k^i \we \tp\left(\K^j - \ga \L^j\right)
\end{equation*}
Similarly to the previous case, making use of \eqref{eq:KL_tangent}, \eqref{eq:naLK_normal}, \eqref{eq:KL_normal} and \eqref{eq:projections_k}, remembering also the projections we can rewrite it as
\begin{align*}
  \ep_{kij} &\left(\tp h^i - \chi^i{}_m \ep^m \right)\we \ep^j =  \ep_{kij} \de_m \ep^i \we \ep^j - \np z^j \we \ep_k \we \ep_j + \tfrac 1 2 \ep_{jlm}\de_m \ep^l_A \ep^{mA} \we \ep_k \we \ep^j  
\end{align*}
Expanding in coordinates and using \eqref{eq:detepsilon} we have
\begin{align*}
    \ep^{kE}\ep^{BCD} \ep_{kij} &\left(\tp h^i{}_B - \chi^i{}_B \right) \ep^j_C  = \ep^{kE}\ep^{BCD}\Big( \ep_{kij}  \ep^j_C \de_m \ep^i_B + \np z_B\ep_{kC} - \tfrac 1 2 \ep_{jlm} \ep^{mA}  \ep^j_B\ep_{kC} \de_m \ep^l_A\Big)
\end{align*}
Rearranging terms we have
\begin{equation}
    \begin{aligned}
    \tfrac 1 2 \bar \ep \ep^{BDE}&\np z_B = \left(\tp h^{DE} - \chi^{DE} \right) - \ga^{ED}\left(\tp h - \chi \right)  -  \ga^{B(E}_{\vphantom{-1} }\ep^{D)}_i \de_m \ep^i_B + \ga^{ED}_{\vphantom{-1} }\ep^{B}_i \de_m \ep^i_B 
    \end{aligned}
    \label{eq:normal_first_eq}
\end{equation}

The skew-symmetric part of \eqref{eq:normal_first_eq} gives
\[\np z^k = \tfrac 1 2 \ep \tp h^{DE} \ep_{BED}\ep^{jB}\]
and, since $\tp h^{[DE]} = 0$ from \eqref{eq:skew_h}, we can write 
\begin{equation}
    \boxed{\np z^k = 0}
    \label{eq:t_z}
\end{equation}

The symmetric part of \eqref{eq:normal_first_eq} is 
\begin{equation}
0 = \left(\tp h^{(DE)} - \chi^{DE}\right) -  \ga^{ED} \left(\tp h - \chi\right) - \ga^{B(E}_{\vphantom i}\ep_i^{D)} \de_m\ep^i_B + \ga^{ED} \ep^B_i \de_m \ep^i_B
\label{eq:sym_normal_first_eq}
\end{equation}
and can be traced
\begin{align*}
   \ep^k_B \de_m \ep^B_k = \tp h - \chi
\end{align*}
so that, substituting back into \eqref{eq:sym_normal_first_eq}, we have 
\begin{align*}
    \ep_{i(E} \de_m \ep_{D)}^i = \tp h_{(DE)} - \chi_{DE}.
\end{align*}
But the extrinsic curvature can be expressed in terms of the triad as in \eqref{eq:extrinsic_curvature_triad} 
and hence
\begin{equation}
     \boxed{\tp h_{(DE)} = 0}
\end{equation}

\subsection{Projecting the third and fourth equations}
Before starting to project the last two equations, let us summarise the results obtained from the first two equations. The independent ones are

\begin{equation}
    \begin{cases}
        DE_k = 0 \hfill  \qquad \qquad &\text{(3 equations)}\\
        \tp z_{(AB)} = 0 &\text{(6 equations)} \\
        \np z^i = 0 &\text{(3 equations)} \\
        \tp h_{[AB]} = 0 &\text{(3 equations)} \\
        \tp h_{(AB)} = 0 &\text{(6 equations)} \\
        \np h^i = 0 &\text{(3 equations)}
    \end{cases}
    \quad \Longleftrightarrow \quad
        \begin{cases}
        DE^k = 0 \hfill  \\
        \tp z_{(AB)} = 0 \\
        \np \A^i = \np \tilde \A^i \\
        k^i = \tilde k^i
    \end{cases}
\end{equation}
which are $21$ algebraic constraints and $3$ differential ones. Let us note that, as expected, we obtained $h = 0$ and hence $k^i = \tilde k$. However, to fully obtain that $\A^i = \tilde \A^i$, we need to solve Gauss constraint. 

Let us notice that at least some of these algebraic constraints appears as a consequence of the definition of the ABI connection on spacetime.
They reproduce some {\it ad hoc} prescriptions, e.g.~$k^i$ expressed as the extrinsic curvature defined by the frame, usually done in the literature, just as a direct consequence of field equations.

\subsubsection{Normal and tangent projections of the third equation}

The normal projection of \eqref{eq:field_3} is 

\begin{align*}
    0 =&  \np F_k \we \ep^k - \tfrac 1 \ga \np \na k^k \we \ep_k +  \ep_{ijk}\, \np k^i \, \tp k^j \we \ep^k =  \\ 
    =& \ep^{ABC}\Big(F^k_A \ep_{kB} - \tfrac 1 \ga  (\np\na k^k)_A \ep_{kB} + \ep_{ijk} \, \np k^i \, \tp k^j_A \, \ep^k_B \Big) =  \\ 
    =& \ep^{ABC}\bar N\bar \ep \Big( \ep_{B}{}^{FE} \left[\chi_{EA} D_F N + ND_F\chi_{EA}\right] +  \ep^{FE}{}_{B} \chi_{EC}  D_AN   \\ 
    &\quad\qquad +  \tfrac 1 \ga \ep \left[\de_m \chi_{BA} + \chi_{BE} \chi^E{}_A + D_{AB}N\right] \Big) = \\
    =& \ep^{ABC}\ep^{FE}{}_{B}\left[\chi_{EA} D_F N + \chi_{EC} D_A N + ND_F\chi_{EA}\right] = \\ 
    =& N \ep^{ABC}\ep^{FE}{}_{B}D_F\chi_{EA}  = D_A(\chi^A{}_C - \de^A_C \chi)
\end{align*}
Thus, we obtain 

\begin{equation}
\boxed{0 = D_A(\chi^A{}_C - \de^A_C \chi)}
\label{eq:projections_3n}
\end{equation}
This equation is identically satisfied with the gauge fixing $N = 1, \be^A = 0$. This is indeed not evident from expression \eqref{eq:projections_3n}

\begin{lma}
    \label{lma:eq_3n_gauge}
    The normal projection of the third equation in \eqref{eq:field_spacetime} is identically satisfied in any gauge fixing in which $\be^A = 0$ is chosen.
\end{lma}
The proof of this lemma can be found in Appendix \ref{appendix:lemma}.

\subsubsection{Tangent projection of the third equation: Hamiltonian constraint}
The tangent projection of \eqref{eq:field_3} is

\begin{align*}
    &\tp\,\F_k\we \ep^k - \tfrac 1 \ga \tp\,\na k^k \we \ep_k + \tfrac 1 2 \ep_{kij} \tp k^i \we \tp k^j \we \ep^k = \tfrac \La 6 \ep_{ijkl} \ep^i \we \ep^j \we \ep^k \\
    &\ep^{BCA}\Bigg(\tfrac 1 2\eta_{ik} F^i_{BC}\ep^k_A + \tfrac 1 \ga \ep^{kD}\ep_{kA}D_B\chi_{DC}+ \tfrac 1 2\ep_{ijk} \ep^{iF}\chi_{FB} \ep^{jG}\chi_{GC}  \ep^k_A     \Bigg) = \tfrac\La 6 \ep_{ijk}\ep^i_B \ep^j_C \ep^k_A\ep^{BCA} \\ 
    & \tfrac 1 4 \ep^{ABC} \ep^k{}_{lm}  \Ri(3)lm_BC. \ep_{kA} + \tfrac 1 \ga D_B \chi_{AC} \ep^{BCA} + \tfrac 1 2 \ep^{BCA}\ep_{kij}\ep^{iF}\ep^{jG}\ep^k_A\chi_{FB}\chi_{GC}= \La \ep\\ 
    & \Ri(3)BC_BC. + \chi^B{}_B\chi^F{}_F - \chi^{CB}\chi_{BC} = 2\La
\end{align*}
Thus, we obtain the \textbf{Hamiltonian constraint}
\begin{equation}
\boxed{\Rs(3) + \chi^2 - \chi_{AB}\chi^{AB} - 2\La  = 0}
\label{eq:hamiltonian_constraint}
\end{equation}

\subsubsection{Tangent projection of the fourth equation: Momentum constraint}
The Momentum constraint follows from the tangent projection of \eqref{eq:field_4}, where the right-hand side is zero because of \eqref{eq:projections_e}
\begin{equation}
     \tfrac 1 \ga \ep^h{}_{ij}\F^i\we \ep^j + e^h{}_{ij}\tp\left(\na k^i\right) \we \ep^j  - \tfrac 1 \ga  \tp k^h \we \tp k_j  \we \ep^j = 0.
\end{equation}

Expressing it in adapted coordinates and using \eqref{eq:projections_3n}, we have
\begin{align*}
    0 &= \ep^{BCA}\ep_h^G\Big( \tfrac 1 {2\ga} \ep^h{}_{ij} F^i_{BC}\ep^j_A - \ep^h{}_{ij} \ep^j_AD_B(\ep^{iD}\chi_{CD}) - \tfrac 1 \ga \ep^{hD}\chi_{DB}\ep^F_j\chi_{FC}\ep^j_A\Big) = \\ 
    &=   \ep^{BCA}\ep_{hij}\ep^j_A\ep^h_G \tfrac 1 {2\ga} F^i_{BC}  = \tfrac 1 {2\ga} \ep (\de^B_G\ep^C_i - \de^C_G\ep^B_i)  F^i_{BC} = \tfrac 1 {2\ga}\ep\left( F^i_{GC}\ep^C_i - F^i_{BG} \ep^B_i\right)
\end{align*} 
and hence we obtain the \textbf{Momentum constraint}
\begin{equation}
\boxed{F^i_{BC} E^C_i = 0}
\label{eq:momentum_constraint}
\end{equation}

\subsubsection{Normal projection of the fourth equation: Evolution equations}
The evolution equations are obtained from the normal projection of the fourth equation in \eqref{eq:field_spacetime}
\begin{align*}
   \tfrac 1 \ga \ep^h{}_{ij}\> \np\F^i\we \ep^j &+ \ep^h{}_{ij}\>\np\! \na k^i\we \ep^j +\tfrac 1 \ga \np k_j \, \tp k^h \we \ep^j - \tfrac 1 \ga \np k^h \, \tp k_j \we \ep^j = \\ 
   &= -\tp\, \F^h + \tfrac 1 \ga \, \tp\,\na k^h - \tfrac 1 2 \ep^h{}_{ij} \tp k^i \we \tp k^j + \tfrac \La 2 \ep^h{}_{ij}\ep^i \we \ep^j)
\end{align*}
which, in adapted coordinates, reads as
\begin{align*}
    \ep^{ABC}\ep^G_h\Big(\tfrac 1 \ga \ep^h{}_{ij} F^i_A \ep^j_B &+ \ep^h{}_{ij} \ \np\! \na k^i_A \ep^j_B  + \tfrac 1 \ga \np k_j \tp k^h_A \ep^j_B - \tfrac 1 \ga  \np k^h \tp k_{jA}\ep^j_B \Big) = \\ 
   &=  \ep^{ABC}\ep^G_h\Big(\! - \tfrac 1 2 F^h_{AB}  + \tfrac 1 \ga \tp \, \na k^h_{AB} -  \tfrac 1 2 \ep^h{}_{ij} \tp k^i_A \tp k^j_B - \tfrac \La 2  \ep^h{}_{ij} \ep^i_A \ep^j_B\Big)
\end{align*}
The terms containing $\tfrac 1 \ga$ expand to
\begin{align*}
    &\tfrac 1 \ga \ep^{ABC}\ep^G_h\left( \ep^h{}_{ij} F^i_A \ep^j_B + \np k_j \tp k^h_A \ep^j_B - \np k^h \tp k_{jA}\ep^j_B - \tp\> \na k^h_{AB} \right) =\\ 
    &= \tfrac 1 \ga \ep^{ABC}\bar N \left( \ep_{hij}\ep_G^h \ep^j_B \ep^i_D \ep^{DEF}\bar \ep D_E(N \chi_{FA}) - (\chi_{GA}D_B N  -  \chi_{BA} D_G N   - ND_A\chi_{GB})\right) = \\
    &= \tfrac 1 \ga \ep^{ABC} \bar N \left[  D_B(N\chi_{GA}) - D_G(N\chi_{BA}) - D_B(N\chi_{GA})\right] = 0
\end{align*}
while the remaining ones give, using \eqref{eq:extrinsic_curvature_triad} together with \eqref{eq:hamiltonian_constraint},
\begin{align*}
    \ep^{ABC}&\ep^G_h\left(\ep^h{}_{ij} \ \np\! \na k^i_A \ep^j_B + \tfrac 1 2 F^h_{AB} + \tfrac 1 2 \ep^h{}_{ij} \tp k^i_A \tp k^j_B + \tfrac \La 2  \ep^h{}_{ij} \ep^i_A \ep^j_B\right) = \\ 
    &= \bar \ep \ep^{ABC}\ep_{GBD}\bar N \left(\ga^{FD}\de_m \chi_{FA} + N\chi^D{}_E\chi^E{}_A + \ga^{FD}D_{AF}N\right) + \\ 
    & \qquad+ \tfrac 1 2 \bar\ep\ep^{ABC}\ep_{GDE}\left(\tfrac 1 2 \Ri(3)DE_AB.  +\chi^D{}_A\chi^E{}_B\right) + \tfrac \La 2\bar \ep \ep^{ABC}\ep_{GAB} = \\
    &=  \bar \ep \bar N \left(\ga^{FC}\de_m\chi_{FG} + \ga^{FC}D_{GF}N  + \de^C_G\left(-\de_m\chi -  \ga^{AF}D_{AF}N\right)\right) + \\ 
    & \qquad + \bar \ep \left(2\chi_{BG}\chi^{BC} - \Ri(3)C_G.  - \chi\chi^C{}_{G}   + \chi^{AF}\chi_{AF}\de^C_G \right)  + \\ 
    & \qquad + \tfrac 1 2 \bar \ep\de^C_G \left( \Rs(3) - \chi^{AE}\chi_{AE} + \chi^2 - 2\La\right) 
\end{align*}

Hence we obtain
\begin{align*}
     &\de_m\chi_{FG} - N \> \Ri(3)_FG. - N\chi\chi_{FG} +  D_{GF}N + 2N\chi_{DG}\chi^{D}{}_{F}    = \\ 
        & \qquad =  \ga_{FG}\left(\de_m\chi+ \gamma^{AD}D_{AD}N    + N\chi^{AD}\chi_{AD}\right)
\end{align*}

The skew-symmetric part of this equation is identically satisfied; before looking at the symmetric part, lets us consider the trace
\begin{align*}
    \de_m\chi - N \chi_{FG}\chi^{FG} + N\La + \ga^{FG}D_{FG}N = 0
\end{align*} 
Replacing it back in the symmetric part, we obtain the \textbf{evolution equation}
\begin{equation}
        \boxed{\de_m\chi_{FG} = N \left( \Ri(3)_FG. + \chi\chi_{FG} - 2\chi_{DG}\chi^D{}_F - \La\right) - D_{FG}N.}
\end{equation}

\section{Summary of the results}

The equations obtained from the projections of the field equations \eqref{eq:field_spacetime} can be summarised as follows.

We have, in total, $10$ differential constraints, $21$ algebraic constraints and $6$ evolution equations, which sum up to $37$ equations. This is exactly the number of field components to be determined: in fact, from the $40$ components of the field in \eqref{eq:holst_equations}, we used $3$ of them for the boosts that adapt the frame, leaving $37$ of them to be determined.
\begin{align*}
    &
    \text{Differential constraints:} &&
    \\
    & \begin{dcases}
        DE^k = 0  \\
        D_A(\chi^A{}_C - \delta^A_C \chi) = 0 \\
        {}^3\!R + \chi^2 - \chi_{AB}\chi^{AB} - 2\Lambda  = 0 \\
        F^i_{BC}E^C_i = 0 
    \end{dcases} 
    && \begin{aligned}
        &&\text{(3 equations)} \\
        &&\text{(3 equations)} \\
        &&\text{(1 equation)} \\
        &&\text{(3 equations)}
    \end{aligned} \\[2ex]
    &\text{Algebraic constraints:} \\ 
    & \begin{dcases}
        \tp z_{(AB)} = 0 \\
        \np z^i = 0 \\
        \tp h_{[AB]} = 0 \\
        \tp h_{(AB)} = 0 \\
        \np h^i = 0 
    \end{dcases} 
    && \begin{aligned}
        &&\text{(6 equations)} \\
        &&\text{(3 equations)} \\
        &&\text{(3 equations)} \\
        &&\text{(6 equations)} \\
       &&\text{(3 equations)}
    \end{aligned} \\[2ex]
    &\text{Evolution equations:} \\ 
    & \quad \de_m\chi_{FG} = N \left( \Ri(3)_FG. + \chi\chi_{FG} - 2\chi_{DG}\chi^D{}_F - \La\right) - D_{FG}N.  
    && \begin{aligned}
        &&\text{(6 equations)}
    \end{aligned}
\end{align*}

It is also noteworthy that the obtained results are consistent with the standard $3+1$ decomposition of the Einstein equations (see e.g. \cite{wald2009, gourgoulhon}) and in particular with the case $\be \neq 0$. Moreover, we were able to obtain the same results without fixing the gauge $N = 1, \be^A = 0$, which is a common choice in the literature, and gave us the opportunity to understand that the $3$ identically satisfied equations that arose from the normal projection of the third equation in \eqref{eq:field_spacetime} were indeed related to the gauge fixing of the lapse and shift functions. 

Finally, having obtained the expected results is not only a consistency check of the calculations, but also opens the possibility to start exploring Holst-like actions in dimensions different from $4$.

One last remaining matter concerns the compatibility of the constraints with the evolution equations. In fact, we have obtained a set of constraints and evolution equations, but it is still to be proven that imposing the constraint on the initial conditions automatically implies that they are satisfied at any later time. In standard GR, this is guaranteed by the Bianchi identities, but it is still to be proven in this case. A viable strategy could be to make use of the conservation laws, of which Bianchi identities are a particular case in GR, and this is something that we plan to investigate in the future.

%\newpage
\appendix
\numberwithin{equation}{subsection}
\renewcommand{\theequation}{\thesection\arabic{subsection}.\arabic{equation}}

\section{Identities and projections of fields }
\label{appendix:projections}

Let us collect here some result used in the projection of the field equations. Let us remark that eventually these are just direct computations. They are important precisely because they show that the procedure to split equations into constranits and evolution equations is somehow algorithmic, even though annoyingly dependent on the model and fields appearing in the model.

\subsection{Projections needed for the first two equations}

Let us here collect the results used to project frames and the first two equations.

\subsubsection{Frame, coframe and their differentials}
Let us start by recalling the projections of the tetrad and cotetrad, which are needed to decompose fields and equations and are found in \eqref{eq:projections_e}
\begin{equation*}
    \tp(e_0) = 0 \quad\qquad \np(e_0) = 1  \quad\qquad \tp (e_i) = \ep_i \quad\qquad  \np(e_i)  = 0.
\end{equation*}
The differential $de^a = d_\mu e^a_\nu dx^\mu \we dx^\nu$ can be projected as 
   \begin{equation}
     \tp(de^0) = 0 \qquad \np(de^0) = - \bar N d_\ast N  \qquad \tp(de^i) = d_\ast \ep^i \qquad  \np(de^i) = \bar N \de_m \ep^i.
   \end{equation}
The operator $\de_m$ controls the evolution of the triad and and is defined to act on the frame as 
\begin{equation}
    \de_m \ep^i_A := d_0\ep^i_A - \be^B d_B \ep^i_A - d_A\be^B\ep^i_B
\end{equation} 

\subsubsection{Extrinsic curvature and Spin connection and}
These results are standard ones and can be found in \cite{gourgoulhon} and, with a notation similar to the one adopted here, in \cite{BookFatibene}  and we will just summarize the needed components here. We will not report the explicit expressions of the projections of $\chr g \al\be\mu$. 

We can express the extrinsic curvature $\chi_{\mu\nu}$ in adapted coordinates as 
\begin{equation}
    \begin{aligned}    
        & \chi_{00} = \be^A\be^B\chi_{AB} \\ 
        & \chi_{0A} = \chi_{A0} = \be^B\chi_{AB} \\
        &\chi_{AB} = - \bar N \left(\tfrac 1 2 d_0 \ga_{AB}- \tilde D_{(A}\be_{B)}\right)
    \end{aligned}
\end{equation}

Moreover, we have
\begin{equation}
 \chi_{AB} = -\tfrac 1 2 \de_m \ga_{AB} = -e_{i(A} \de_m e_{B)}^i 
 \label{eq:extrinsic_curvature_triad}
\end{equation}
where $\de_m\ga_{AB}$ is computed from $\de_m e^i_A$, since $\ga_{AB} = e^i_A e^j_B \de_{ij}$, and is given by
\begin{equation}
    \de_m \ga_{AB} = d_0 \ga_{AB} + \ga^{AD}D_B\be^B + \ga^{BD}D_D\be^A.
    \label{eq:de_m_metric}
\end{equation}

Consider the Christoffel symbols of $g$, $\chr g \al\be\mu$, and the ones of $\ga$, $\chr \ga AB C$. We can compare their expressions to the ones of the extrinsic curvature and obtain the following relations
\begin{subequations}
\begin{align}
        &\chr g ABC = \chr \ga ABC + \bar N \be^A \chi_{CB} \\ 
        &\chr g 0BC = - \bar N \chi_{BC}      \\
        &\chr g 0B0 = \bar N \left(D_B N - \chi_{BD}\be^D\right)\\ 
        &\chr g AB0 = -\be^A \bar N (D_B N - \chi_{DB} \be^D) + D_B\be^A - N\chi^A{}_B
    \end{align}
\end{subequations}
and, by inversion, 
\begin{subequations}
    \begin{align}
        &\chi_{BC} = - N \chr g 0BC \\
        &\chr \ga ABC = \chr g ABC +  \be^A \chr g 0BC \\ 
        &\chr g 0B0 = \bar N D_B N + \chr g 0BE \be^E.
    \end{align}
    \label{eq:christoffel_extrinsic_curvature}
\end{subequations}

To compute the projections of $\tilde A$ and $\tilde k$, we need to project the frame-compatible spin connection $\tilde \om^{ab}$, defined so that $\tilde\na_\mu e^a_\nu = 0$, namely 
\begin{equation}
    \tilde \om  = e^a_\al \left(\chr g \al\be\mu e^\be_c + d_\mu e^\al_c\right)
\end{equation} 

The resulting projections of the spin connection are
\begin{subequations}
    \begin{align}
        &\tp\tilde\om^{0k} = -\ep^{Bk}\chi_{BA} dx^A && \np \tilde\om^{0k} = \ep^{Bk}\bar N D_B N \\ 
        &\tp\tilde\om^{ij} =: \chr\ep ikA dx^A && \np \tilde\om^{ik} = \ep^{iC}\bar N\de_m \ep^{Ck} - \ep^i_C\ep^{Bk}\chi^{C}{}_{B} 
    \end{align}
    \label{eq:projections_spin_connection}
\end{subequations}

\subsubsection{Connection and curvature }

The projections of the connection $\A^i$  and of the Immirzi tensor $k^i$ are 
\begin{subequations}
    \begin{align}
        \tp(\A^i) &=:  A^i_C dx^C  & \np (\A^i) &= \bar N \left(A^i_0 - \be^C A^i_C\right) =:  a^i \\ 
        \tp(k^i) &=: k^i_C dx^C   & \np (k^i) &= \bar N \left(k^i_0 - \be^C k^i_C\right) 
    \end{align}
    \label{eq:tmp_projections_A_k}
\end{subequations}
Since we are interested in expressing the results in terms of the differences $z^i = \A^i - \tilde \A^i$ and $h^i = k^i - \tilde k^i$, we also need the projections of $\tilde \A^i$ and $\tilde k^i$. 

From \eqref{eq:variables_change} and     \eqref{eq:projections_spin_connection} we get
\begin{subequations}
    \begin{align}
        &\tp \tilde\A^i_A = \tfrac 1 2 \ep^i{}_{jk}\{\ep\}^{jk}_A \qquad \qquad && \np \tilde\A^i = - \tfrac 1 2 \bar N \ep^i{}_{jk}\de_m \ep^j_C \ep^{kC} \\
        &\tp \tilde k^i_A = - \ep^{iB}\chi_{BA}   \qquad \qquad
        && \np \tilde k^i = \ep^{iB} \bar N D_B N
    \end{align}
\end{subequations}
so that we can express the projections of $\A$ and $k$ in terms of the projections of $\tilde A$, $\tilde k$, $z$ and $h$ in the followng way
\begin{subequations}
    \begin{align}
        &A^i_A = \tp z^i_A + \tfrac 1 2 \ep^i{}_{jk}\{\ep\}^{jk}_A \qquad \quad & &
         a^i = \np z^i - \tfrac 1 2 \bar N \ep^i{}_{jk}\de_m \ep^j_C \ep^{kC}     \label{eq:projections_A} \\
        & \tp k^i_A = \tp h^i_A - \ep^{iB}\chi_{BA}   \qquad \quad
        & & \np k^i = \np h^i + \ep^{iB} \bar N D_B N     \label{eq:projections_k}
    \end{align}
\end{subequations}

The projections of the covariant derivative $\na k^k := dk^k + \ep^k{}_{ij} \A^i \we k^j $  are
\begin{subequations}
    \begin{align}
        \tp\,(\na k^i) &= D_A \tp k^i_B dx^A\we dx^B = -D_A\left(\ep^{iD}\chi_{DB}\right)dx^A \we dx^B \label{eq:naktangent}\\ 
        \np\,(\na k^i) &= \begin{aligned}[t]
            \Big[- \ep^{iD}\left(d_0 \chi_{DA} + \chi_{DB} \chi^B{}_A\right) + \bar N\big(&\ep^i_D\be^CD_C\chi^D{}_A \\ &- D_A\left(\ep^{iC}D_CN\right)\big)\Big] dx^A \\ 
        \end{aligned} \\ 
        &= \left[- \ep^{iD}\left(\de_n \chi_{DA} + \chi_{DB} \chi^B{}_A\right) - \bar N D_A\left(\ep^{iC}D_CN\right)\right] dx^A
    \end{align}
    \label{eq:projections_na_A_k}
\end{subequations}

Finally, we can express the covariant derivative of the frame in terms of the difference $z^i$ as follows
\begin{lma}
    \[D_B \ep^k_C = \ep^k {}_{il} \tp z^l_B \ep^i_C\]
    where $D_B$ the projection of the covariant derivative with respect of $\A^i$ on $S_t$. \label{lma:covder_ep}
\end{lma}
\begin{proof}
\begin{align*}
    D_B \ep^k_C &= d_B\ep^k_C + \ep^k{}_{il}A^l_B \ep^i_C=  \Big(\tilde\na_B \ep^k_C + \left(\ep^k{}_{il}A^l_B - \tilde \om^k{}_{iB}\right)\ep^i_C\Big) = \nonumber \\
    &= \ep^k{}_{il}\left(A^l_B - \tilde A^l_B\right)\ep^i_C = \ep^k{}_{il} \tp z^l_B\ep^i_C = \ep \ep^F_{\> \cdot  \> CG} \ep^k_F \ep^G_l \tp z^l_B 
\end{align*}    
where we used the fact that $\tilde \na$ is the covariant derivative with respect to the spin connection induced by the triad, which is compatible with it, i.e. $\tilde \na_B \ep^k_C = 0$.
\end{proof}

\subsubsection{Other projections}

We also need the projections of the the two linear combinations of $\K^i$ and $\L^i$ 
\[
    \L^k + \ga \K^k = \tfrac{1 + \ga^2}{2\ga} \ep^k{}_{ij} e^i \we e^j \qquad \K^k - \ga \L^k = - \tfrac{1 + \ga^2}{\ga} e^0 \we e^k
\]%
\begin{subequations}%
    \renewcommand{\theequation}{\theparentequation\alph{equation}}%
    \begin{align}
        & \tp(\L^i + \ga \K^i) = \tfrac{1 + \ga^2}{2\ga} \ep^i{}_{jk} \ep^j \we \ep^k = \tfrac{1 + \ga^2}\ga E^k \label{eq:LK_tangent}\\
        & \np(\L^i + \ga \K^i) =  0 \label{eq:LK_normal}\\ 
        &\tp(\K^i - \ga\L^i) = 0 \label{eq:KL_tangent}\\ 
        & \np(\K^i - \ga\L^i) = -\tfrac{1 + \ga^2}{\ga} \ep^i\label{eq:KL_normal}
        \end{align}
\end{subequations}

To compute the projections of their covariant derivatives
\begin{align*}
    \na \left(\L^k + \ga \K^k\right) &:= d\left(\L^k + \ga \K^k\right) + \ep^k{}_{ij} \A^j \we \left(\L^i + \ga \K^i\right) \\
    \na \left(\K^k - \ga \L^k\right) &:= d\left(\K^k - \ga \L^k\right) + \ep^k{}_{ij} \A^j \we \left(\K^i - \ga \L^i\right)
\end{align*}
we first need the projections of their differentials 
 \[
    d\left(\L^k + \ga \K^k\right) = \tfrac{1 + \ga^2}{\ga} \ep^k{}_{ij} de^i \we e^j \qquad d\left(\K^k - \ga \L^k\right) = - \tfrac{1 + \ga^2}{\ga} \left( de^0 \we e^k + e^0 \we de^k\right).
\]

For these, we have
\begin{subequations}
    \renewcommand{\theequation}{\theparentequation\alph{equation}}
    \begin{align}
        \tp\left(d\left(\L^k + \ga \K^k\right)\right) &= \tfrac{1 + \ga^2}\ga \ep^k{}_{ij} d_\ast \ep^i \we \ep^j \label{eq:dLK_tangent}\\ 
        \tp\left(d\left(\K^k - \ga \L^k\right)\right) &= 0\label{eq:dKL_tangent}\\ 
        \np\left(d\left(\L^k + \ga \K^k\right)\right) &= \tfrac{ 1 + \ga^2}\ga \ep^k{}_{lm} \de_n \ep^l \we \ep^m \label{eq:dLK_normal}\\ 
        \np\left(d\left(\K^k - \ga \L^k\right)\right) &= \tfrac{1 + \ga^2}\ga \left(\bar N d_\ast N \we \ep^k - \de_n \ep^k\right)\label{eq:dKL_normal}
    \end{align}
\end{subequations}
and therefore we have 
\begin{subequations}
    \renewcommand{\theequation}{\theparentequation\alph{equation}}
    \begin{align}
        &\tp\left(\na \left(\L^k + \ga \K^k\right)\right) = \tfrac{1 + \ga^2}\ga\ep^k{}_{ij} D \ep^i \we \ep^j  \label{eq:naLK_tangent}\\
        &\tp\left(\na \left(\K^k - \ga \L^k\right)\right) = 0  \label{eq:naKL_tangent} \\
        &\np\left(\na \left(\L^k + \ga \K^k\right)\right) = \tfrac{1 + \ga^2}\ga \left( \ep^k{}_{ij} \de_n \ep^i \we \ep^j - a^j \we e^k \we e_j\right) \label{eq:naLK_normal} \\ 
        &\np\left(\na \left(\K^k - \ga \L^k\right)\right) = -\tfrac{1 + \ga^2}\ga \left(\bar N d_\ast N \we \ep^k + D\ep^k\right)\label{eq:naKL_normal}
    \end{align}
\end{subequations}

\subsection{Specialising projections for the third and fourth equations}

Before projecting the third and fourth equations in \eqref{eq:field_spacetime}, it is useful to rewrite some of the previous found results after imposing the algebraic constraints. We stress that it is just a matter of rewriting the results via direct computations and substitutions.

First of all, since $z = 0$ we have $A^i_A = \tilde A^i_A$ and hence from \eqref{eq:projections_A} 
\begin{equation}
    A^i_A = \tfrac 1 2 \ep^i{}_{jk}\{\ep\}^{jk}_A \qquad a^i = - \tfrac 1 2 \bar N \ep^i{}_{jk}\de_m \ep^j_C \ep^{kC}
    \label{eq:projections_A_z=0}
\end{equation}
From  \eqref{eq:projections_k} we have
\begin{subequations}
\begin{align}
   &\tp k^i = k^i_A dx^A = - \ep^{iB} \chi_{BA} dx^A \\
   & \np k^i = \ep^{iB} \bar N D_B N  \end{align}    
\end{subequations}
and from \eqref{eq:projections_na_A_k}
\begin{subequations}
\begin{align}
   & \tp\,(\na k^i)_{AB} = -e^{iD} D_A\chi_{DB}  \\  
   & \np(\na k^k)_A = - \bar N e^{kD}\left(\de_m \chi_{DA} + N\chi_{DE} \chi^E{}_A + D_{AD}N\right).
\end{align}    
\end{subequations}

Since the third and fourth equations in \eqref{eq:field_spacetime} also involve the curvature of the connection $\A^i$, namely $\F^i = \tfrac 1 2 \big( d_\mu \A^i_\nu - d_\nu \A^i_\mu - \ep^i{}_{jk}  \A^j_\mu \A^k_\nu \big)$,  before continuing we need to write its projections  
\begin{equation}
    \tp(\F^i) := \tfrac 1 2 F^i_{BC} dx^B\we dx^C \qquad \qquad  \np(\F^i) := F^i_A dx^A 
\end{equation}

where
    \begin{align*}
         F^i_{BC} &= d_B A^i_C - d_C A^i_B - \ep^i{}_{jk} A^j_B A^k_C  = \\ 
        &= \tfrac 1 2 \ep^i{}_{jk} d_B \{e\}^{jk}{}_C - \tfrac 1 2 \ep^i_{jk} d_C \{e\}^{jk}{}_B - \ep^i{}_{jk} \left(\tfrac 1 2 \ep^j{}_{lm}\{e\}^{lm}{}_B\right)\left(\tfrac 1 2 \ep^k{}_{pq}\{e\}^{pq}{}_C\right) = \\ 
      \Hide{  = \tfrac 1 2 \ep^i{}_{jk} \left(d_B \{e\}^{jk}{}_C - d_C \{e\}^{jk}{}_B + \{e\}^{j}{}_{lB} \{e\}^{lk}{}_C - \{e\}^{j}{}_{lC} \{e\}^{lk}{}_B\right) = }
        &= \tfrac 1 2 \ep^i{}_{jk} \Ri(3)jk_BC.
    \end{align*}
        \begin{align*}
            F^i_A &= \bar N \left[\left(d_0 A^i_A - d_A A^i_0 - \ep^i{}_{jk} A^j_0 A^k_A\right) - \be^C F^i_{CA} \right] = \\ 
            &= \bar N \left( d_0 A^i_A - D_A A^i_0 - \be^C F^i_{CA}\right) = \bar N \left[ d_0 A^i_A - D_A \left(Na^i + \be^C A^i_C \right) - \be^C F^i_{CA}\right] 
        \end{align*} 

Now, using \eqref{eq:projections_A_z=0} and Lemma \ref{lma:covder_ep} and the definition of $\de_m$ we can express 
\begin{align*}
    D_A(Na^i + \be^C A^i_C) &= -\tfrac 1 2 \ep^i{}_{jk}D_A\left(d_0 \ep^j_C - \be^E D_E \ep^j_C - D_E \be^E \ep^i_C\right)\ep^{kC} = \\ 
    &=  -\tfrac 1 2 \ep^i{}_{jk}\ep^{kC}\left[D_A(d_0 \ep^j_C)- D_A(\be^E D_E \ep^j_C )- D_A(D_C \be^E \ep^j_E)\right]  = \\ 
    \Hide{= -\tfrac 1 2 \ep^{ijk} \ep_k^C D_A(d_0 \ep_{jC})  + \tfrac 1 2 \ep^{ijk} \ep_k^C \ep_j^E D_A D_C \be_E =} 
    \Hide{= \begin{aligned}[t]
        -\tfrac 1 2 \ep^{ijk} \ep^C_k &\left[d_0(D_A \ep_{jC}) - \ep_{jml}\ep^m_C d_0 A^l_A + \ep_{jE} d_0\chr \ga E AC\right]  \\ & + \tfrac 1 2 \ep^{ijk} \ep_k^C \ep_j^E D_A D_C \be_E =
    \end{aligned}
        }
    &= d_0 A^i_A + \tfrac 1 2 \ep^{ijk} \ep_{jE}\ep_k^C d_0 \chr \ga E AC + \tfrac 1 2 \ep^{ijk} \ep_k^C \ep_j^E D_A D_C \be_E
\end{align*}
so that 
\begin{align*}
F^i_A &= - \tfrac 1 2 \bar N\ep^{ijk} \left[  \ep_{jE}\ep_k^C d_0 \chr \ga E AC + \ep_k^C \ep_j^E D_A D_C \be_E + \be^C \Ri(3)_jkCA.\right] = \\ 
&= -\tfrac 1 2 \bar N \ep^{ijk}\Big\{\begin{aligned}[t] &\ep_{jE} \ep_k^C  d_0\left[\tfrac 1 2 \ga^{EF} \left(d_A \ga_{FC} + d_C \ga_{FA} - d_F \ga_{AC}\right)\right] + \\\ & + \ep_k^C \ep_j^E D_A D_C \be_E + \be^D \Ri(3)_jkDA.\Big\}\end{aligned}=
\\ 
    \Hide{= \tfrac 1 2\bar N \bar \ep \ep^i_B \ep^{BCE}  \Big[\begin{aligned}[t]- d_0&\ga_{ED} \chr \ga D AC + \tfrac 1 2  \left(d_A d_0 \ga_{EC} + d_C d_0 \ga_{EA} - d_E d_0 \ga_{AC}\right) + \\ &+  D_AD_C\be_E +  \be^D \Ri(3)_ECDA.\Big] \end{aligned}=} 
    &= \tfrac 1 2\bar N \bar \ep \ep^i_B \ep^{BCE}   \Big[\begin{aligned}[t]- d_0\ga_{ED}& \chr \ga D AC + d_C d_0 \ga_{EA} + d_0 \ga_{AD} \chr \ga D CE  + \\ &+ D_AD_C\be_E +  \be^D \Ri(3)_ECDA.\Big]\end{aligned} = \\
    &=  \tfrac 1 2\bar N \bar \ep \ep^i_B \ep^{BCE}  \left[ -D_C d_0 \ga_{EA} + D_AD_C\be_E +  \be^D \Ri(3)_ECDA.\right] = \\ 
    &= \tfrac 1 2\bar N \bar \ep \ep^i_B \ep^{BCE} \left[- D_C\left(d_0 \ga_{EA} \, \pm\,  D_A\be_E \pm D_E \be_A\right)  + D_AD_C\be_E +  \be^D \Ri(3)_ECDA.\right] = \\ 
     \Hide{= \bar N \bar \ep \ep^i_B \ep^{BCE} \left[D_C\left(N \chi_{EA}\right) + \tfrac 1 2 \left(-D_C D_A \be_E - D_C D_E \be_A + D_AD_C\be_E +  \be^D \Ri(3)_ECDA.\right)\right] = }
    \Hide{ = \bar N \bar \ep \ep^i_B \ep^{BCE} \left[D_C\left(N \chi_{EA}\right) + \tfrac 1 2 \Big( - D_{CE}\be_A  +\tfrac 1 2  \be^D \Ri(3)_DACE. +  \be^D \Ri(3)_DECA. + \be^D  \Ri(3)_ECDA. \Big) \right] = }
    &=  \bar N \bar \ep \ep^i_B \ep^{BCE} D_C\left(N \chi_{EA}\right)
\end{align*}

\section{Proof of Lemma \ref{lma:eq_3n_gauge}}
\label{appendix:lemma}
The proof is based on the following lemma
\begin{lma}
    \[D_A \chi_{BD} \ep^{ABC} =  \ep^{ABC} \bar N \be^E \left(\Ri(3)_E BDA. - \chi_{EA} \chi_{BD}\right)\]
    where $D_A$ is the projection of the covariant derivative with respect of $\A^i$ on $S_t$,
\end{lma}
\begin{proof}    

By direct computation, exploiting the symmetry of $\chi_{AB}$ and making use of \eqref{eq:christoffel_extrinsic_curvature}, we have
\begin{align*}
D_A \chi_{BD}\ep^{ABC} =& \left(-D_D\chi_{BA} + D_A\chi_{BD}\right)\ep^{ABC} \\
=& \left( D_D \left(N \chr g0BC \right)  - D_A \left(N \chr g 0BD\right)\right) \\  
 \Hide{= \begin{aligned}[t]\Big[\big(& \chr g 0BC D_D N - \chr g0BD D_AN\big) + N\Big( -d_D\chr g0BA - \chr g 0ED \chr\ga EBA + \\ & - \chr g 0BE \chr\ga EAD + d_A \chr g0BD +\chr g 0EA\chr\ga EBD  + \chr g 0BE\chr\ga EDA \Big)\Big]\ep^{ABC}  =\end{aligned}} 
 \Hide{ = \begin{aligned}[t] \ep^{ABC}\Big[&\big(\chr g 0BC D_D N - \chr g0BD D_AN\big) + N\Big(d_D \chr g0BA - d_A \chr g0BD + \\ &+ \chr g0ED \left(\chr gEBA + \chr g0BA\be^E\right) - \chr g0EA\left(\chr gEBD + \chr g0BD\be^E\right)\Big)\Big] =  \end{aligned}}
=&   \begin{aligned}[t]\ep^{ABC}\Big[&\big(\chr g 0BC D_D N - \chr g0BD D_AN\big)  + \\ &+ N \Big(d_D \chr g0BA - d_A \chr g0BD + \underbrace{\chr g0ED\chr gEBA+  \chr g00D \chr g0BA}_{\chr g0\al D \chr g\al BA} +\\ &-   \chr g00D\chr g0BA \underbrace{ - \chr g0EA\chr gEBD -   \chr g00A\chr g0BD}_{- \chr g0\al A \chr g\al BD} +\\ &+   \chr g 00A \chr g 0BD  + \be^E\left(\chr g0ED \chr g 0BA - \chr g 0EA \chr g 0BD \right)\Big)\Big]= \end{aligned}\\
  \Hide{ = \begin{aligned}[t]\ep^{ABC}\Big[&\big(\chr g 0BC D_D N - \chr g0BD D_AN\big) + N \Big(\Ri(4)0_BDA.  + \be^E\big(\chr g 0ED\chr g 0BA + \\ &- \chr g 0EA\chr g0BD \big) - \chr g 00D \chr g 0BA + \chr g 00A \chr g 0BD \Big)\Big]=\end{aligned}}
=& N \ep^{ABC} \,  \Ri(4) 0_BDA.  
\end{align*} 

Next, we can express $\Ri(4)0_BDA.$ in terms of the Riemann tensor of the spatial metric $\ga$, as defined in \ref{eq:projections_metric}, and of the extrinsic curvature $\chi_{AB}$ as follows:
\begin{align*}
    \Ri(4)0_BDA. &= g^{0\al}\ \Ri(4)_\al BDA. = \bar N^2\left(  \Ri(4)_C BDA. \be^C - \Ri(4)_0BDA. \right) = \\ &= \bar N^2\left( g_{\al C}\Ri(4)\al_BDA.\be^C  - \Ri(4)_0BDA.  \right) =  \\ &= \bar N^2\left( g_{EC}\Ri(4)E_BDA. \be^C + g_{0C}\Ri(4)0_BDA. \be^C - \Ri(4)_0BDA.  \right) = \\ &= \bar N^2\left( \ga_{EC}\Ri(4)E_BDA. \be^C + \ga_{EC}\be^E \be^C \Ri(4)0_BDA. - \Ri(4)_0BDA.  \right)
\end{align*}

Hence  
\begin{equation*}
     \Ri(4)0_BDA. =\tfrac 1 {N^2 - |\be|^2} \left( \be_E \Ri(4)E_BDA. - \Ri(4)_0BDA.\right)
\end{equation*}

The second term vanishes because of (first) Bianchi identities
\begin{align*}
    \ep^{ABC} \ \Ri(4)_0BDA. &= \ep^{0ABC}\  \Ri(4)_DAB0. = \\ &= \tfrac 1 3 \left(\Ri(4)_DAB0. \ep^{AB0C}+ \Ri(4)_DB0A. \ep^{B0AC} + \Ri(4)_D0AB. \ep^{0ABC}\right) = \\ 
    & = \tfrac 1 3 \Ri(4)_C\al\be\mu. \ep^{\al\be\mu C} = 0
\end{align*}

We can express the first term in terms of the Riemann tensor of the spatial metric $\ga$ 
\begin{align*}
    \Ri(4)E_BDA. &= d_D \chr g EBA + \chr g E D \al \chr g \al BA - [D\leftrightarrow A]= \\ 
    &= d_D\left(\chr \ga EBA + \bar N \be^E\chi_{BA}\right) + \\ 
    & \qquad - \bar N\Big( \!-\be^E \bar N (D_D N - \chi_{FD} \be^F) + D_D\be^E - N\chi^E{}_D\Big)\chi_{BA} + \\ 
    & \qquad + \left(\chr \ga E D F + \bar N \be^E \chi_{DF}\right)\left(\chr \ga F B A + \bar N \be^F \chi_{BA}\right) - [D\leftrightarrow A] = \\
    \Hide{ =\left( d_D\chr \ga EBA + \chr \ga E D F \chr \ga F B A - [D\leftrightarrow A]\right) + \underline{d_D\left(\bar N \be^E \chi_{BA}\right)} + }
     \Hide{\qquad - \dunderline{d_A\left(\bar N \be^E \chi_{BD}\right) } + \left(\be^E \bar N^2 (D_D N - \chi_{FD} \be^F) - \bar N D_D\be^E + \chi^E{}_D\right)\chi_{BA}  } 
     \Hide{ \qquad + \bar N^2 \be^E \be^F \chi_{DF} \chi_{BA}  + \left(\be^E \bar N^2 (D_A N - \chi_{FA} \be^F) - \bar N D_A\be^E + \chi^E{}_A\right)\chi_{BD} + }
    \Hide{ \qquad + \underline{\bar N \chr \ga EDF \be^F \chi_{BA}} +\dunderline{ \bar N \be^E \chi_{DF } \chr \ga FBA} - \dunderline{\bar N \chr \ga EAF \be^F \chi_{BD}} + }
    \Hide{\qquad - \underline{\bar N \be^E \chi_{AF} \chr \ga FBD} +  \underline{\bar N \be^E \chi_{BF}\chr \ga FDA} - \dunderline{\bar N \be^E \chi_{DF} \chr \ga FBA}  =} 
    &=  \Ri(3)E_BDA. + D_D\left(\bar N \be^E \chi_{BA}\right) - D_A\left(\bar N \be^E \chi_{BD}\right) - D_D(\bar N\be^E)  \chi_{BA} + \\
    &  \qquad + D_A(\bar N \be^E)\chi_{BD}   + \chi^E{}_D \chi_{BA} - \chi^E{}_A \chi_{BD} \\ 
    &= \Ri(3)E_BDA. + \bar N \be^E D_D \chi_{BA} - \bar N \be^E D_A \chi_{BD} + \chi^E{}_D \chi_{BA} - \chi^E{}_A \chi_{BD} 
\end{align*}

Hence 
\begin{align*}
    D_A \chi_{BD} \ep^{ABC}   \Hide{= \tfrac N {N^2 - |\be|^2} \ep^{ABC} \be_E \big(\begin{aligned}[t]
        &\Ri(3)E_BDA. + \bar N \be^E D_D \chi_{BA} - \bar N \be^E D_A \chi_{BD} + \\ &+ \chi^E{}_D \chi_{BA} - \chi^E{}_A \chi_{BD} \big) 
    \end{aligned}= }
    &= \tfrac 1 {N^2 - |\be|^2} \ep^{ABC} \left(N \be^E \ \Ri(3)_E BDA. - |\be|^2 D_A \chi_{BD} - N \be^E \chi_{EA} \chi_{BD}\right)
\end{align*}

and finally 

\begin{align*}
     D_A \chi_{BD} \ep^{ABC} = \ep^{ABC} \bar N \be^E \left(\Ri(3)_E BDA. - \chi_{EA} \chi_{BD}\right)
\end{align*}

\end{proof}

This lemma shows that the right hand side of \eqref{eq:projections_3n}
is  identically zero if $\be^A = 0$.

\paragraph{Acknowledgements} This paper is also supported by INdAM-GNFM. We also acknowledge the contribution of INFN (Iniziativa Specifica QGSKY
and Iniziativa Specifica Euclid), the local research project Metodi Geometrici
in Fisica Matematica e Applicazioni (2025) of Dipartimento di Matematica of
University of Torino (Italy).

 \nocite{*}
\printbibliography

\end{document}